\documentclass[spanish, english]{article}
\usepackage{layout}
\usepackage{babel}
\usepackage[OT1]{fontenc}
\usepackage{graphicx}
\usepackage{latexsym}
\usepackage{amsmath}
\usepackage{amssymb}
\usepackage{amsthm}
\usepackage[mathscr]{eucal}
\usepackage{array}
\usepackage{amsfonts}%
\usepackage{centernot}
\usepackage{multirow}
\usepackage{babel}
\usepackage{enumerate}
\usepackage{amsbsy}
\usepackage{epsfig}

\decimalpoint

\newcommand{\keyclave}[1]{\textbf{\textit{Palabras Clave -- }} #1}
\newcommand{\keyw}[1]{\textbf{\textit{Keywords -- }} #1}

\newtheorem{theorem}{Teorema}
\newtheorem{prop}{Proposici\'on}
\newtheorem{lemma}{Lema}
\newtheorem{corollary}{Corolario}

\def\W2{W^{1,2}({\cal O}(M))}

\def\1half{\frac{1}{2}}

\newcommand{\beq}{\begin{equation*}}
\newcommand{\eeq}{\end{equation*}}

\newcommand{\be}{\begin{equation}}
\newcommand{\ee}{\end{equation}}

\newcommand{\Prob}{\mathbb{P}}
\newcommand{\Mean}{\mathbb{E}}
\newcommand{\R}{\mathbb{R}}

\newcommand{\anos}{a$\tilde{\text{n}}$os\ }

\newcommand{\ene}{$\tilde{\text{n}}$}

\begin{document}

\title{Modelaci\'on de Poblaciones V\'ia\\ Cadenas de Markov Tridimensionales\\ (Demographic Modeling Via 3-dimensional Markov Chains)}
%
%
%
\author{Juan Jos\'e V\'iquez\footnote{Profesor/Investigador en el Departamento de Matem\'aticas de la Universidad de Costa Rica, email: viquezejin@gmail.com.}, Jorge Aurelio V\'iquez\footnote{Graduado del Departamento de Matem\'aticas de la Universidad de Costa Rica, email: javiquez42@gmail.com.}, Alexander Campos\footnote{Graduado del Departamento de Matem\'aticas de la Universidad de Costa Rica, email: alexander.camp353@gmail.com.}, Jorge Lor\'ia\footnote{Graduado del Departamento de Matem\'aticas de la Universidad de Costa Rica, email: jelorias95@gmail.com.}, Luis Alfredo Mendoza\footnote{Graduado del Departamento de Matem\'aticas de la Universidad de Costa Rica, email: luis.mf08@gmail.com.}.}

%
%
%
%
%
\maketitle

\selectlanguage{english}
\begin{abstract}
This article presents a new model for demographic simulation which can be used to forecast and estimate the number of people in pension funds (contributors and retirees) as well as workers in a public institution. Furthermore, the model introduces opportunities to quantify the financial flows coming from future populations such as salaries, contributions, salary supplements, employer contribution to savings/pensions, among others. The implementation of this probabilistic model will be of great value in the actuarial toolbox, increasing the reliability of the estimations as well as allowing deeper demographic and financial analysis given the reach of the model. We introduce the mathematical model, its first moments, and how to adjust the required probabilities, showing at the end an example where the model was applied to a public institution with real data.
\end{abstract}
\keyw{Markov Chains, Demographic Simulation, Financial Engineering}

\selectlanguage{spanish}

\begin{abstract}
En este art\'iculo se presenta un nuevo modelo de generaci\'on poblacional que puede ser utilizado para proyectar tanto personas en fondos de pensiones (tanto cotizantes como jubilados) como trabajadores en instituciones p\'ublicas. Aunado a esto, el modelo presenta oportunidades para cuantificar los flujos derivados de estas poblaciones futuras, tales como gastos en salarios, cotizaciones, pluses salariales, aportes patronales a ahorros/pensiones, entre otros. Claramente la implementaci\'on de este modelo probabil\'istico ser\'a de gran utilidad dentro de la caja de herramientas actuariales, aumentando la confiabilidad de las proyecciones, as\'i como permitiendo an\'alisis m\'as profundos por cuanto el desgloce poblacional y financiero del modelo es extenso. Aqu\'i presentamos el modelo matem\'atico, sus primeros momentos, y el ajuste de las probabilidades que lo alimenta, finalizando con un ejemplo de aplicaci\'on a una instituci\'on p\'ublica con datos reales.
\end{abstract}
\keyclave{Cadenas de Markov, Generaci\'on Demogr\'afica, Matem\'atica Financiera}

\section{Introducci\'on}

A menudo en el mundo actuarial, se presenta la imperiosa necesidad de contar con proyecciones poblacionales. Para reg\'imenes de pensiones del tipo ``Pay-As-You-Go'', es necesario pronosticar el n\'umero de cotizantes que alimenten los ingresos del fondo, as\'i como la cantidad de pensionados (quienes representan los gastos). Del mismo modo, en una instituci\'on p\'ublica, es de gran relevancia conocer la din\'amica de la movilida laboral, analizando el crecimiento/decrecimiento de trabajadores en ciertos puestos, salidas por jubilaciones y relevos generacionales, al tiempo que se cuantifican los gastos derivados de estos trabajadores y sus pluses salariales.

\bigskip

El siguiente modelo se sustenta en considerar una tripleta markoviana, es decir, un vector en el espacio de estados $\mathcal{E}\subset \R^3$, tal que para todo conjunto $\{w_i\ /\ i=1,\dots,n\}\subset\mathcal{E}$ se debe cumplir
$$\Prob\left[\boldsymbol{W}_n = w_n\ |\ \boldsymbol{W}_1 = w_1,\dots,\boldsymbol{W}_{n-1} = w_{n-1}\right]=\Prob\left[\boldsymbol{W}_n = w_n\ |\ \boldsymbol{W}_{n-1} = w_{n-1}\right].$$
Se descompone el espacio $\mathcal{E}$ en tres componentes importantes para la modelaci\'on: un espacio de ``Categor\'ia'', uno de ``Edad'' y otro de ``Antig$\ddot{\text{u}}$edad''. La idea detr\'as del mismo consiste en considerar que esta tripleta determina el comportamiento del individuo, y los estados a los que ``salta''.

\bigskip

Por ejemplo, considere la situaci\'on de modelar el comportamiento de una universidad, espec\'ificamente para la categor\'ia acad\'emica de catedr\'atico. Personas en esta categor\'ia pueden tener distintas edades, donde claramente un catedr\'atico con 35 a\~nos va a acceder a pluses salariales distintos que a los que accede un catedr\'atico de 50 a\~nos, el cual puede ser incluso Rector de la universidad. A\'un m\'as, suponiendo que estamos observando a un catedr\'atico de 35 a\~nos, existen distintas formas que esto suceda. Por ejemplo, el trabajador podr\'ia haber logrado el estatus de catedr\'atico en otra entidad y llevar solo un a\~no trabajando en esta universidad, o podr\'ia llevar 15 a\~nos ``asociado'' a esta instituci\'on. Se esperar\'ia que el primero presente un comportamiento muy distinto al segundo, el cual ha construido una carrera profesional dentro de la instituci\'on y le ser\'ia m\'as dif\'icil salirse.

\bigskip

El art\'iculo est\'a dividido de la siguiente manera: en la segunda secci\'on se presentan unos resultados b\'asicos de cadenas de Markov que ser\'an de utilidad en las siguientes partes. A lo largo de la tercera secci\'on se presenta la definici\'on probabil\'istica del modelo, sus estados y sus propiedades, presentando resultados sobre poblaciones esperadas y c\'omo calcularlas. Para la cuarta secci\'on se tratar\'a el tema del ajuste estad\'istico del modelo, apalancandose en datos mensuales para computar los estimadores de las probabilidades de transici\'on, de ingreso al sistema, y de la distribuci\'on inicial. La quinta secci\'on se concentra en presentar un algoritmo eficiente de generaci\'on poblacional, utilizando el hecho de que se puede considerar cada grupo poblacional como una realizaci\'on de una multinomial con par\'ametros determinados por el modelo. Finalmente, en la sexta secci\'on se presentar\'a la implementaci\'on del modelo con datos reales, y se mostrar\'a el nivel de ajuste al contrastarlo con datos observados, utilizando el conocido m\'etodo del ``backtesting''.

\section{Preliminares}

Antes de iniciar con la presentaci\'on del modelo de generaci\'on poblacional, es necesario establecer unos resultados conocidos\footnote{Ver \cite{Hoel} para adentrarse m\'as en el tema.}, que ser\'an herramientas \'utiles para efectuar los c\'alculos necesarios del modelo.

\begin{lemma}\label{probcondicional1}
Si $D_i$ son disjuntos y $\Prob[C\ |\ D_i]=p$, independientemente de $i$, entonces $\Prob[C\ |\ \cup_iD_i]=p$.
\end{lemma}

\begin{proof}
$$\Prob[C\ |\ \cup_iD_i]=\frac{\Prob\left[C\cap \cup_iD_i\right]}{\Prob[\cup_iD_i]}=\frac{\sum_i \Prob[C\cap D_i]}{\Prob[\cup_kD_k]}=\frac{1}{\Prob[\cup_kD_k]}\sum_i \overbrace{\Prob[C\ |\ D_i]}^{=\ p}\Prob[D_i]=p\frac{1}{\Prob[\cup_kD_k]}\overbrace{\sum_i \Prob[D_i]}^{=\ \Prob[\cup_iD_i]}=p.$$
\end{proof}

\begin{lemma}\label{probcondicional2}
Si $C_i$ son disjuntos, entonces $\Prob[\cup_iC_i\ |\ D]=\sum_i\Prob[C_i\ |\ D]$.
\end{lemma}

\begin{proof}
Por aditividad contable de $\Prob$ se tiene que
$$\Prob[\cup_iC_i\ |\ D]=\frac{\Prob[(\cup_iC_i)\cap D]}{\Prob[D]}=\frac{\Prob[\cup_i(C_i\cap D)]}{\Prob[D]}=\sum_i\frac{\Prob[C_i\cap D]}{\Prob[D]}=\sum_i\Prob[C_i\ |\ D].$$
\end{proof}

\begin{lemma}\label{probcondicional3}
Si $E_i$ son disjuntos y $\cup_iE_i=\Omega$, entonces $\Prob[C\ |\ D]=\sum_i\Prob[E_i\ |\ D]\Prob[C\ |\ E_i\cap D]$.
\end{lemma}

\begin{proof}
Note que
$$\Prob[E_i\ |\ D]\Prob[C\ |\ E_i\cap D]=\frac{\Prob[E_i\cap D]}{\Prob[D]}\cdot\frac{\Prob[C\cap E_i\cap D]}{\Prob[E_i\cap D]}=\Prob[C\cap E_i\ |\ D].$$
Como los $C\cap E_i$ son disjuntos, entonces por el Lema \ref{probcondicional2},
$$\sum_i\Prob[E_i\ |\ D]\Prob[C\ |\ E_i\cap D]=\sum_i\Prob[C\cap E_i\ |\ D]=\Prob[\cup_i(C\cap E_i)\ |\ D]=\Prob[C\cap (\cup_iE_i)\ |\ D]=\Prob[C\ |\ D].$$
\end{proof}

\section{Modelo Probabil\'istico}

En las siguientes secciones se definir\'an los estados que componen al espacio $\mathcal{E}$, se efectuar\'an los c\'alculos sobre la cadena una vez implementados ciertos supuestos, y se concluir\'a con el ajuste estad\'istico de los par\'ametros del modelo. Adem\'as, se crear\'an particiones en $N-$tuplas de cada tripleta, agregando caracter\'isticas poblacionales importantes dentro de la modelaci\'on.

\subsection{Definici\'on de los Estados}

Se considera un espacio de estados $\mathcal{E}=C\times E\times A$, compuesto por un espacio de ``Categor\'ia'' $C$, uno de ``Edad'' $E$ y otro de ``Antig$\ddot{\text{u}}$edad'' $A$, con estados determinados por la forma $(c,e,a)\in C\times E\times A$, siendo $c$ la categor\'ia a la que pertenece, $e$ la edad que posee, y $a$ la antig\"uedad que tiene asignada.

\begin{itemize}

\item {\bf Categor\'ia:} Asumimos que existen $N_C+1$ categor\'ias, es decir, $C=\bigl\{C^{(0)},C^{(1)},\dots,C^{(N_C)}\bigr\}$. Aqu\'i cada $C^{(i)}$ representa alg\'un tipo de indicador sobre el estatus del individuo, tales como categor\'ia salarial en caso de trabajador p\'ublico, o tipo de sector y g\'enero (hombre/mujer - independiente/p\'ublico/privado). Por otra parte, $C^{(0)}$ es la categor\'ia que representa estar ``fuera'' del sistema, es decir, son las personas que no pertenecen a la organización, y que no le están generando ningún tipo de gastos o ingresos (directamente), pero que con probabilidad positiva pueden llegar a hacerlo en los siguientes años. No se asume ninguna estructura sobre estos, únicamente se busca la probabilidad de que entren a cada categoría. Por ejemplo, en una instituci\'on p\'ublica se considerar\'ia como ``no estar contratado'' por dicha entidad, o como personas no cotizantes si fuera un fondo de pensiones.

\bigskip

\item {\bf Edad:} Se toman edades enteras (aunque se pueden desagregar m\'as) como un rango $E=[E_l,E_u)\cap\mathbb{Z}$, donde $E_l$ es la edad m\'as peque\~na (posiblemente negativa) y $E_u$ la edad m\'axima. Se consideran $N_E+1$ grupos de edad, es decir, $E=\bigcup_{i=0}^{N_E} E^{(i)}$. La idea es buscar grupos de edad que presenten similares comportamientos de transici\'on entre categor\'ias. El grupo de edad $E^{(0)}$ representa las edades ``de reserva'', es decir, aquellas personas que con el paso de los a\~nos vendr\'an a alimentar el modelo. Por ejemplo, si estamos en una instituci\'on p\'ublica, la cual contrata solamente a personas mayores de 18 a\~nos, y si se va a proyectar la poblaci\'on por 25 a\~nos, entonces este grupo de reserva ser\'ia 
$$E^{(0)}=[-7,18)\cap\mathbb{Z}.$$

\item {\bf Antig$\ddot{\text{u}}$edad:} Del mismo modo, se consideran solamente antig\"uedades enteras $A=[0,A_u)\cap\mathbb{Z}$, las cuales representan los a\~nos de ``ligamen'' del individuo con el ``sistema''.\footnote{Se puede desagregar más, pero no se consideran antig\"uedades así en este art\'iculo debido a que se complica mucho la notación (ya suficientemente compleja).} Tomamos $N_A$ grupos de antig$\ddot{\text{u}}$edades, con $A=\bigcup_{i=1}^{N_A} A^{(i)}$. En un fondo de pensiones, los grupos de antig$\ddot{\text{u}}$edades ser\'ian ``paquetes'' de a\~nos cotizados, indicando su proximidad/lejan\'ia con el estado de la categor\'ia``pensi\'on''. Por ejemplo, $A^{(1)}$ podr\'ia estar compuesto por las personas que tienen entre 1 a 10 a\~nos de cotizar. En una instituci\'on p\'ublica, se podr\'ia tomar a las antig\"uedades como a\~nos laborados en dicha instituci\'on, indicando la ``carrera'' profesional que la persona haya construido en dicha entidad.

\end{itemize}

\smallskip

Sea $S_i$ la $i-$\'esima caracter\'istica, con $S_i:=\left\{S_i^{(1)},S_i^{(2)},\dots,S_i^{(N_i)}\right\}$, donde $S_i^{(j)}$ es el $j-$\'esimo estado de la $i-$\'esima caracter\'istica. Tomamos el vector $S=(S_1,S_2,\dots,S_{N_S})$ de dichas caracter\'isticas, el cual representa aquellos factores que inciden en los c\'alculos, financieros o demogr\'aficos, de la poblaci\'on a la cual se le est\'a aplicando el modelo. Se denota por $I_n^{C^{(r)},E^{(i)},A^{(k)}}\bigl(S^{j_1}_1,\dots,S^{j_{N_S}}_{N_S}\bigr)$ al n\'umero de personas que en el a\~no $n$ poseen una tripleta $(c,e,a)\in \left\{C^{(r)}\right\}\times E^{(i)}\times A^{(k)}\subset\mathcal{E}$, y $N_S-$tupla de caracter\'isticas $\bigl(S^{j_1}_1,\dots,S^{j_{N_S}}_{N_S}\bigr)$. 

\bigskip

\noindent {\bf Ejemplo:} Tomemos el caso de un profesor de la Universidad de Costa Rica, para el cual se tienen los rubros salariales usuales, (se incluyen los montos de las garant\'ias sociales). El c\'alculo del salario de cada uno de los grupos de la poblaci\'on se har\'ia de acuerdo a los valores de $S^{j_i}_i$, correspondientes a los siguientes componentes:

\begin{table}[htbp]
\begin{center}
\begin{tabular}{lclcr}\hline
$S^{j_1}_i$ & & Descripci\'on  & & Monto \\\hline
$S^{j_1}_1$ & = &	Salario Base Docente 	&  &	644 831	\\
$S^{j_2}_2$	& = &	Porcentaje Categor\'ia Acad\'emica	& &	354 657	\\
$S^{j_3}_3$	& = &	Anualidad	&  &	776 190	\\
$S^{j_4}_4$	& = &	Escalaf\'on Docente 	&  &	119 939	\\
$S^{j_5}_5$	& = &	Fondo Consolidado 	&  &	18 854	\\
$S^{j_6}_6$	& = &	Pasos Acad\'emicos	&  &	59 970	\\
$S^{j_7}_7$	& = &	Reconocimiento por Elecci\'on 	&  &	279 857	\\
$S^{j_8}_8$	& = &	Magisterio	&  &	176 822	\\
$S^{j_9}_9$	& = &	Seguro de Enfermedad y Maternidad	&  &	225 600	\\
$S^{j_{10}}_{10}$	& = &	Banco Popular	&  &	12 195	\\
$S^{j_{11}}_{11}$	& = &	Fondo de Capitalizaci\'on Laboral	&  &	73 168	\\
$S^{j_{12}}_{12}$	& = &	Fondo de Pensi\'on Complementaria	&  &	36 584	\\
$S^{j_{13}}_{13}$	& = &	Aguinaldo	&  &	203 235	\\
$S^{j_{14}}_{14}$	& = &	Salario Escolar	&  &	184 627	\\
$S^{j_{15}}_{15}$	& = &	JAFAP	& &	60 973	\\\hline
  &  &	TOTAL GASTADO	& &	3 227 500\\\hline
\end{tabular}
\end{center}
\vspace{-.5cm}
\caption{Caracter\'isticas Salariales}
\label{CaracteristicasSalariales}
\end{table}

\smallskip

Como se observa de los datos, el costo total para la Universidad por este profesor es de $3.227.500$ colones. Asuma ahora que para $n=2$,
$$I_2^{C^{(r)},E^{(i)},A^{(k)}}\bigl(S^{j_1}_1,\dots,S^{j_{15}}_{15}\bigr)=30,$$
es decir, que dentro de 2 a\~nos habr\'ian 30 individuos en la categor\'ia $C^{(r)}$, con rango de edad $E^{(i)}$, rango de antig\"{u}edad $A^{(k)}$, y con las caracter\'isticas salariales del Cuadro
\ref{CaracteristicasSalariales}. Entonces, el gasto total para la Universidad por este grupo de personas ser\'a de $96.824.998$ colones. Repitiendo este proceso con todas las categor\'ias, rangos de edad y de antig\"uedad, y todas las caracter\'isticas salariales, se logra obtener el monto total gastado en todos los empleados de la Universidad de Costa Rica.\footnote{El salario base se deber\'ia incrementar (semestralmente) de acuerdo con la inflaci\'on estimada, para cada a\~no de proyecci\'on.}

\subsection{Cadena de Markov}

Considere, $(X_n,Y_n,Z_n)$ la tripleta aleatoria de la cadena, donde $X_n$, $Y_n$ y $Z_n$ representan el estado ``Categor\'ia'', ``Edad'' y ``Antig$\ddot{\text{u}}$edad'', respectivamente, en el $n-$\'esimo a$\tilde{\text{n}}$o. Se asume que se cumple la propiedad de Markov, es decir, para $(c_k,e_k,a_k)\in C\times E\times A$, $k=0,\dots,n$,
\begin{align*}
\Prob\bigl[(X_n,Y_n,Z_n)=\left(c_n,e_n,a_n\right)\ \bigr|&\ (X_{n-1},Y_{n-1},Z_{n-1})=\left(c_{n-1},e_{n-1},a_{n-1}\right),\dots, (X_0,Y_0,Z_0)=\left(c_0,e_0,a_0\right)\bigr]\\
&=\Prob\bigl[(X_n,Y_n,Z_n)=\left(c_n,e_n,a_n\right)\ \bigr|\ (X_{n-1},Y_{n-1},Z_{n-1})=\left(c_{n-1},e_{n-1},a_{n-1}\right)\bigr].
\end{align*}

\noindent Como se asume que la cadena es homog\'enea, se tiene que
\begin{align*}
\Prob\bigl[(X_n,Y_n,Z_n)=\left(c_n,e_n,a_n\right)\ \bigr|&\ (X_{n-1},Y_{n-1},Z_{n-1})=\left(c_{n-1},e_{n-1},a_{n-1}\right)\bigr]\\
&=\Prob\bigl[(X_1,Y_1,Z_1)=\left(c_1,e_1,a_1\right)\ \bigr|\ (X_0,Y_0,Z_0)=\left(c_0,e_0,a_0\right)\bigr],
\end{align*}
para todo $n$. Se nota adem\'as que
\begin{align*}
\Prob\bigl[(X_1,Y_1,Z_1)=\left(c_1,e_1,a_1\right)\ \bigr|&\ (X_0,Y_0,Z_0)=\left(c_0,e_0,a_0\right)\bigr]\\
&=\Prob\bigl[X_1=c_1\ |\ Y_1=e_1,Z_1=a_1,X_0=c_0,Y_0=e_0,Z_0=a_0\bigr]\\
&\quad\quad\quad \times\Prob\bigl[Z_1=a_1\ |\ Y_1=e_1,X_0=c_0,Y_0=e_0,Z_0=a_0\bigr]\\
&\quad\quad\quad\quad\quad \times\Prob\bigl[Y_1=e_1\ |\ X_0=c_0,Y_0=e_0,Z_0=a_0\bigr]
\end{align*}

\smallskip

{\bf Observaciones, hip\'otesis y c\'alculos:}
\begin{itemize}

\item Asumimos que el incremento de la antig$\ddot{\text{u}}$edad es homog\'enea (identicamente distribuida); eso es, existe una variable aleatoria $\xi$ con valores en $\left\{0,1\right\}$ (que simboliza el haber pertenecido al sistema (``trabajado-cotizado'') durante un a$\tilde{\text{n}}$o o no haber estado en el sistema en dicho a$\tilde{\text{n}}$o), tal que $Z_1-Z_0\overset{d}{=}\xi$, para todo $n$.
\begin{align*}
\Prob\bigl[X_1=c_1\ |&\ Y_1=e_1,Z_1=a_1,X_0=c_0,Y_0=e_0,Z_0=a_0\bigr]\\
&=\Prob\bigl[X_1=c_1\ |\ Y_1=e_1, \xi=a_1-a_0,X_0=c_0,Y_0=e_0,Z_0=a_0\bigr]
\end{align*}
y
\begin{align*}
\Prob\bigl[Z_1=a_1\ |\ Y_1=e_1,X_0=c_0,&Y_0=e_0,Z_0=a_0\bigr]\\
&=\Prob\bigl[\xi=a_1-a_0\ |\ Y_1=e_1,X_0=c_0,Y_0=e_0,Z_0=a_0\bigr]
\end{align*}

\item Se asume que los aumentos de edad suceden en enero de cada a\ene o. N\'otese que si $e_1\neq e_0+1$, entonces $\{Y_1=e_1\}\cap\{Y_0=e_0\}=\emptyset$. Como $\Prob\bigl[Y_1=e_1\ |\ X_0=c_0,Y_0=e_0,Z_0=a_0\bigr]=0$, en este caso no importar\'ia si no se condiciona por $\{Y_1=e_1\}$. Igualmente, $\{Y_1=e_1\}\cap\{Y_0=e_0\}=\{Y_0=e_0\}$ si $e_1= e_0+1$. Y en el segundo caso, condicionar por $\{Y_1=e_1\}\cap\{Y_0=e_0\}$ es lo mismo que condicionar solamente por $\{Y_0=e_0\}$. En conclusi\'on, se tiene que
\begin{align*}
\Prob\bigl[X_1=c_1\ |&\ Y_1=e_1, \xi=a_1-a_0,X_0=c_0,Y_0=e_0,Z_0=a_0\bigr]\\
&=\Prob\bigl[X_1=c_1\ |\ \xi=a_1-a_0,X_0=c_0,Y_0=e_0,Z_0=a_0\bigr],
\end{align*}
y
\begin{align*}
\Prob\bigl[\xi=a_1-a_0\ |&\ Y_1=e_1,X_0=c_0,Y_0=e_0,Z_0=a_0\bigr]\\
&=\Prob\bigl[\xi=a_1-a_0\ |\ X_0=c_0,Y_0=e_0,Z_0=a_0\bigr],
\end{align*}
siempre y cuando se multiplique el factor $\Prob\bigl[Y_1=e_1\ |\ X_0=c_0,Y_0=e_0,Z_0=a_0\bigr]$. A\'un m\'as,
\begin{align*}
\Prob\bigl[Y_1=e_1\ |\ X_0=c_0,Y_0=e_0,Z_0=a_0\bigr]=\begin{cases}
1 & \text{ si } e_1=e_0+1\\
0 & \text{ si } e_1\neq e_0+1
\end{cases}.
\end{align*}

\item Note que la probabilidad de estar en la categor\'ia $c_1\neq C^{(0)}$ dado que no aument\'o su antig\"uedad ($\xi=0$) es 0 (no puede cambiar de categor\'ia dentro del sistema si no est\'a en el sistema). Del mismo modo, la probabilidad de irse a la categor\'ia $C^{(0)}$ si no permanece en el sistema es 1. Inversamente, si entra al sistema ($\xi=1$), entonces la probabilidad de pasar por la categor\'ia $C^{(0)}$ ser\'ia 0, y solo podr\'ia pasar por a las otras categor\'ias $\{C^{(1)},\dots,C^{(N_C)}\}$.\footnote{Recuerde si sale del sistema entonces no puede tener ninguna categor\'ia dentro del sistema, y si entra al sistema no podr\'ia tener la categoria que significa estar afuera del sistema.} De este modo tenemos,
\begin{align*}
\Prob\bigl[X_1=c_1\ |\ \xi=a_1-a_0,X_0=c_0,&Y_0=e_0,Z_0=a_0\bigr]\\
&=\begin{cases}
\Prob\bigl[X_1=c_1\ |\ X_0=c_0,Y_0=e_0,Z_0=a_0\bigr] & \text{ si } \xi = 1 \text{ y } c_1\neq C^{(0)},\\
0 & \text{ si } \xi = 1 \text{ y } c_1 = C^{(0)}\\
1 & \text{ si } \xi = 0 \text{ y } c_1 = C^{(0)}\\
0 & \text{ si } \xi = 0 \text{ y } c_1 \neq C^{(0)}
\end{cases}
\end{align*}

\item Se asume que la probabilidad de pasar de la categor\'ia $c_0$ a la categor\'ia $c_1$, dado que fue contratado o no ($\xi$), $a_0\in A_0\in\left\{A^{(1)},A^{(2)},A^{(3)},A^{(4)}\right\}$, y $e_0\in E_0\in\left\{E^{(0)},E^{(1)},E^{(2)},E^{(3)},E^{(4)}\right\}$, es la misma. Es decir, que la probabilidad de cambiar de categor\'ia solo se ve afectada cuando se pasa de grupos de edad y antig$\ddot{\text{u}}$edad. De la misma manera, $\Prob\bigl[\xi=\cdot\ |\ X_0=c_0,Y_0=e_0,Z_0=a_0\bigr]$ no cambia para todo $a_0\in A_0\in\left\{A^{(1)},A^{(2)},A^{(3)},A^{(4)}\right\}$, y $e_0\in E_0\in\left\{E^{(0)},E^{(1)},E^{(2)},E^{(3)},E^{(4)}\right\}$. Por el Lema \ref{probcondicional1}, se concluye que
$$\hspace{-.2cm}\Prob\bigl[X_1=c_1\ |\ \xi=a_1-a_0, X_0=c_0,Y_0=e_0,Z_0=a_0\bigr]=\Prob\left[X_1=c_1\ \left|\ \xi=a_1-a_0, X_0=c_0,Y_0\in E_0,Z_0\in A_0\right.\right].$$
y
$$\Prob\bigl[\xi=a_1-a_0\ |\ X_0=c_0,Y_0=e_0,Z_0=a_0\bigr]=\Prob\left[\xi=a_1-a_0\ \left|\ X_0=c_0,Y_0\in E_0,Z_0\in A_0\right.\right]$$

\end{itemize}

Denote
$$P^{E_0,A_0}\bigl(c_0,c_1\bigr):=\Prob\left[X_1=c_1\ \left|\ X_0=c_0,Y_0\in E_0,Z_0\in A_0\right.\right],$$
la probabilidad de transici\'on entre las categor\'ias $\{C^{(1)},\dots,C^{(N_C)}\}$, dado los grupos de edad $E_0$ y de antig$\ddot{\text{u}}$edad $A_0$. Igualmente, denote
$$Q^{E_0,A_0,c_0}(\cdot):=\Prob\left[\xi=\cdot\ \left|\ X_0=c_0,Y_0\in E_0,Z_0\in A_0\right.\right],$$
la probabilidad de que una persona en el grupo de edad $E_0$, con antig$\ddot{\text{u}}$edad en el rango $A_0$, dentro de la categor\'ia $c_0$, sea inclu\'ido o no al sistema en el a\~no siguiente. N\'otese que $Q^{E_0,A_0,c_0}(r)=0$ para todo $r\notin\left\{0,1\right\}$.

\bigskip

Tomando en cuenta todo lo anterior, se concluye que si $e_0\in E_0$ y $a_0\in A_0$, entonces
\begin{align}\label{probtrans}
\hspace{-.5cm}\Prob\bigl[(X_1,Y_1,Z_1)=\left(c_1,e_1,a_1\right)\ \bigr|&\ (X_0,Y_0,Z_0)=\left(c_0,e_0,a_0\right)\bigr]\nonumber\\
&=\begin{cases}
P^{E_0,A_0}\bigl(c_0,c_1\bigr)\cdot Q^{E_0,A_0,c_0}(1)\cdot\boldsymbol{1}_{\{e_1=e_0+1\}} & \text{ si } a_1-a_0 = 1 \text{ y } c_1\neq C^{(0)},\\
0 & \text{ si } a_1-a_0 = 1 \text{ y } c_1 = C^{(0)}\\
Q^{E_0,A_0,c_0}(0)\cdot\boldsymbol{1}_{\{e_1=e_0+1\}} & \text{ si } a_1-a_0 = 0 \text{ y } c_1 = C^{(0)}\\
0 & \text{ si } a_1-a_0 = 0 \text{ y } c_1 \neq C^{(0)},
\end{cases}
\end{align}

\subsubsection{Transiciones Mensuales}

Aqu\'i se est\'a abusando del lenguaje, pues un a\~no consiste de 12 meses, y en cada mes se podr\'ia observar una categor\'ia diferente. En este sentido, es necesario considerar el sentido que tiene la frase ``en el $n-$\'esimo a\~no la persona tuvo la tripleta $(c_n,e_n,a_n)$''. Una alternativa es considerar que $(X_n,Y_n,Z_n)=(c_n,e_n,a_n)$ representa que al inicio del a\~no $n$, la persona se encontraba en esa tripleta, y asumir que no existen cambios de categor\'ia, edad y antig\"uedad a lo largo del a\~no. Otra opci\'on ser\'ia pensar que $(X_n,Y_n,Z_n)=(c_n,e_n,a_n)$ significa que la categor\'ia fue $c_n$, la edad $e_n$ y la antig$\ddot{\text{u}}$edad $a_n$, en al menos un mes del $n-$\'esimo a$\tilde{\text{n}}$o. 

\begin{itemize}

\item Asuma que la transici\'on entre las categor\'ias $\{C^{(1)},\dots,C^{(N_C)}\}$ (dado el grupo de edad y de antig$\ddot{\text{u}}$edad) dentro de un mismo a\~no es una cadena de Markov en escala ``mensual''. Es decir, sea $\widetilde{X}_m$ la variable que representa el estado de categor\'ia en que se encuentra en el $m-$\'esimo mes, y sea la probabilidad $\Prob^{E_0,A_0}[\cdot]=\Prob[\ \cdot\ |\ Y_0\in E_0,Z_0\in A_0]$, entonces,
$$\Prob^{E_0,A_0}\left[\widetilde{X}_m=\widetilde{c}_m\ \left|\ \widetilde{X}_{m-1}=\widetilde{c}_{m-1},\dots,\widetilde{X}_0=\widetilde{c}_0\right.\right]=\Prob^{E_0,A_0}\left[\widetilde{X}_m=\widetilde{c}_m\ \left|\ \widetilde{X}_{m-1}=\widetilde{c}_{m-1}\right.\right].$$

\item Como los aumentos de edad y antig\"uedad suceden en enero de cada a\ene o, de manera que la probabilidad de cambiar de categor\'ia mensualmente es la misma durante todo el a\ene o, i.e., la cadena $\widetilde{X}_m$ es homog\'enea para $1\leq m\leq 12$. Denote
$$\widetilde{P}^{E_0,A_0}\bigl(\widetilde{c}_0,\widetilde{c}_1\bigr):=\Prob\left[\widetilde{X}_1=\widetilde{c}_1\ \left|\ \widetilde{X}_0=\widetilde{c}_0,Y_0\in E_0,Z_0\in A_0\right.\right],$$
la probabilidad de transici\'on {\bf mensual} entre categor\'ias, dado los grupos de edad $E_0$ y antig$\ddot{\text{u}}$edad $A_0$.

\item Sea $T_i$ el tiempo de parada donde se alcanza por primera vez la categor\'ia $C^{(i)}$ dentro de un a\~no determinado, con funci\'on de masa de probabilidad $p^{(i)}_t=\Prob[T_i=t]$. Como la tripleta $(X_n,Y_n,Z_n)$ representa el estado de categor\'ia, edad y antig$\ddot{\text{u}}$edad en el $n-$\'esimo a\~no, y de \'estos solo la categor\'ia cambia dentro de un mismo a\~no, definimos la siguiente relaci\'on, con $c_1 = C^{(i)}$,
\begin{align*}
P^{E_0,A_0}\bigl(c_0,c_1\bigr)&:=\Mean^{T_i}\left[\Prob^{E_0,A_0}\left[\widetilde{X}_{T_i}=\widetilde{c}_1\ \left|\ \widetilde{X}_0=\widetilde{c}_0, T\right.\right]\right]\\
&=\sum_{t=1}^{12}\Prob^{E_0,A_0}\left[\widetilde{X}_t=\widetilde{c}_1\ \left|\ \widetilde{X}_0=\widetilde{c}_0\right.\right]\cdot p^{(i)}_t.
\end{align*}
A\'un m\'as, aplicando la propiedad de Markov junto con el Lema \ref{probcondicional3} (dado que $\cup_{\widetilde{c}_i\in C}\{\widetilde{X}_i=\widetilde{c}_i\}=\Omega$) iteradamente, con la convenci\'on de que $c_t=c_1$,
\begin{align}\label{probtransmens}
P^{E_0,A_0}\bigl(c_0,c_1\bigr)&=\sum_{t=1}^{12}\Prob^{E_0,A_0}\left[\widetilde{X}_t=\widetilde{c}_1\ \left|\ \widetilde{X}_0=\widetilde{c}_0\right.\right]\cdot p^{(i)}_t\nonumber\\
&=\sum_{t=1}^{12}\sum_{\widetilde{c}_1,\dots,\widetilde{c}_{t-1}\in C}\prod_{l=1}^{t}\Prob^{E_0,A_0}\left[\widetilde{X}_l=\widetilde{c}_l\ \left|\ \widetilde{X}_{l-1}=\widetilde{c}_{l-1}\right.\right]\cdot p^{(i)}_t\nonumber\\
&=\sum_{t=1}^{12}\sum_{\widetilde{c}_1,\dots,\widetilde{c}_{t-1}\in C}\prod_{l=1}^{t}\widetilde{P}^{E_0,A_0}\bigl(\widetilde{c}_l,\widetilde{c}_{l-1}\bigr)\cdot p^{(i)}_t.
\end{align}
Este t\'ermino se sustituir\'ia en $\eqref{probtrans}$, para calcular las probabilidades de transici\'on.

\end{itemize}

\noindent Este tipo de herramienta se torna \'util cuando se poseen pocos datos anuales, y se requieren utilizar datos mensuales para incrementar las observaciones y mejorar la convergencia de los estimadores de los par\'ametros del modelo.

\subsection{Distribuci\'on de $(X_n,Y_n,Z_n)$}

Se define la distribuci\'on inicial $\pi_{c_0,e_0,a_0}$, para una tripleta $(c_0,e_0,a_0)\in C\times E\times A$, como la probabilidad de que un individuo tenga categor\'ia $c_0$, edad $e_0$ y antig$\ddot{\text{u}}$edad $a_0$ en el a\~no inicial.

\bigskip

\begin{prop}\label{prop1}
La probabilidad de que una persona, despu\'es de un a$\tilde{\text{n}}$o, se encuentre en el estado $(c_1,e_1,a_1)\in C\times E\times A$, es decir, que al a$\tilde{\text{n}}$o siguiente un individuo tenga $e_1$ a$\tilde{\text{n}}$os de edad, con antig$\ddot{\text{u}}$edad $a_1$ y en la categor\'ia $c_1$, vendr\'ia dado por
\begin{align*}
\Prob\bigl[(X_1,Y_1,Z_1)=(c_1,e_1,a_1)\bigr]=\sum_{c_0\in C}\left(P^{E_{e_1-1},A_{a_1-\delta_{c_1}}}\bigl(c_0,c_1\bigr)\right)^{\delta_{c_1}}\cdot Q^{E_{e_1-1},A_{a_1-\delta_{c_1}},c_0}\bigl(\delta_{c_1}\bigr)\cdot\pi_{c_0,e_1-1,a_1-\delta_{c_1}}
\end{align*}
donde $E_e\in\bigl\{E^{(0)},E^{(1)},E^{(2)},E^{(3)},E^{(4)}\bigr\}$ y $A_a\in\bigl\{A^{(1)},A^{(2)},A^{(3)},A^{(4)}\bigr\}$ son tales que $e\in E_e$ y $a\in A_a$, y con
\begin{align*}
\delta_{c_1}=\begin{cases}
1 & \text{ si } c_1 \neq C^{(0)}\\
0 & \text{ si } c_1 = C^{(0)}
\end{cases}.
\end{align*}
\end{prop}

\bigskip

\begin{proof}
Como $\bigcup_{c_0\in C,e_0\in E,a_0\in A}\bigl\{(X_0,Y_0,Z_0)=(c_0,e_0,a_0)\bigr\}=\Omega$, entonces
\begin{align*}
\Prob\bigl[&(X_1,Y_1,Z_1)=(c_1,e_1,a_1)\bigr]=\Prob\left[\bigl\{(X_1,Y_1,Z_1)=(c_1,e_1,a_1)\bigr\}\bigcap\bigcup_{c_0\in C,e_0\in E,a_0\in A}\bigl\{(X_0,Y_0,Z_0)=(c_0,e_0,a_0)\bigr\}\right]\\
&=\sum_{c_0\in C,e_0\in E,a_0\in A}\Prob\bigl[(X_1,Y_1,Z_1)=(c_1,e_1,a_1)\ \bigr|\ (X_0,Y_0,Z_0)=(c_0,e_0,a_0)\bigr]\cdot\Prob\bigl[(X_0,Y_0,Z_0)=(c_0,e_0,a_0)\bigr]\\
&=\sum_{c_0\in C,e_0\in E,a_0\in A}\Prob\bigl[(X_1,Y_1,Z_1)=(c_1,e_1,a_1)\ \bigr|\ (X_0,Y_0,Z_0)=(c_0,e_0,a_0)\bigr]\cdot\boldsymbol{1}_{\{e_1=e_0+1\}}\cdot\pi_{c_0,e_0,a_0}\\
&=\sum_{c_0\in C,a_0\in\left\{a_1-1,a_1\right\}}\Prob\bigl[(X_1,Y_1,Z_1)=(c_1,e_1,a_1)\ \bigr|\ (X_0,Y_0,Z_0)=(c_0,e_0,a_0)\bigr]\cdot\pi_{c_0,e_1-1,a_0}\\
&=\begin{cases}
\sum_{c_0\in C}Q^{E_{e_1-1},A_{a_1},c_0}(0)\cdot\pi_{c_0,e_1-1,a_1} & \text{ si } c_1=C^{(0)}\\
\\
\sum_{c_0\in C}P^{E_{e_1-1},A_{a_1-1}}\bigl(c_0,c_1\bigr)\cdot Q^{E_{e_1-1},A_{a_1-1},c_0}(1)\cdot\pi_{c_0,e_1-1,a_1-1} & \text{ si } c_1\neq C^{(0)}
\end{cases}
\end{align*}
\end{proof}

\begin{theorem}\label{distXYZn}
La probabilidad de que una persona se encuentre en el estado $(c_n,e_n,a_n)\in C\times E\times A$ en el a$\tilde{\text{n}}$o $n$, es decir, que en $n$ a$\tilde{\text{n}}$os un individuo tenga $e_n$ a$\tilde{\text{n}}$os de edad, con antig$\ddot{\text{u}}$edad $a_n$ y en la categor\'ia $c_n$, vendr\'ia dado por
\begin{align*}
\Prob&\bigl[(X_n,Y_n,Z_n)=(c_n,e_n,a_n)\bigr]\\
&=\sum_{c_0,c_1,\dots,c_{n-1}\in C}\prod_{r=1}^n\left(P^{E_{e_n-r},A_{a_n-\sum_{l=1}^r \delta_{c_{n+1-l}}}}\bigl(c_{n-r},c_{n+1-r}\bigr)\right)^{\delta_{c_{n+1-r}}}\cdot Q^{E_{e_n-r},A_{a_n-\sum_{l=1}^r \delta_{c_{n+1-l}}},c_{n-r}}(\delta_{c_{n+1-r}})\\
&\hspace{4cm}\times\pi_{c_0,e_n-n,a_n-\sum_{l=0}^{n-1} \delta_{c_{l+1}}}.
\end{align*}
\end{theorem}

\bigskip

\begin{proof}
Se probar\'a por inducci\'on:\\
El caso $n=1$ es la Proposici\'on \ref{prop1}.\\
Asuma que se cumple para $n-1$, y se probar\'a para $n$. Repitiendo los pasos de la prueba de la Proposici\'on \ref{prop1}, pero tomando $n$ en lugar de $1$ y $n-1$ en lugar de $0$, se obtiene,
\begin{align}\label{eq1}
\hspace{-.2cm}\Prob\bigl[(X_n,Y_n,Z_n)=&(c_n,e_n,a_n)\bigr]\\
\hspace{-.2cm}&=\sum_{c_{n-1}\in C}\left(P^{E_{e_n-1},A_{a_n-\delta_{c_n}}}\bigl(c_{n-1},c_n\bigr)\right)^{\delta_{c_n}}\cdot Q^{E_{e_n-1},A_{a_n-\delta_{c_n}},c_{n-1}}\bigl(\delta_{c_n}\bigr)\nonumber\\
\hspace{-.2cm}&\quad\quad\quad\quad\quad\quad\quad\quad\quad\times\Prob\bigl[(X_{n-1},Y_{n-1},Z_{n-1})=(c_{n-1},e_n-1,a_n-\delta_{c_n})\bigr].\nonumber
\end{align}
Por hip\'otesis de inducci\'on,
\begin{align}\label{eq2}
\Prob&\bigl[(X_{n-1},Y_{n-1},Z_{n-1})=(c_{n-1},e_{n-1},a_{n-1})\bigr]\\
&=\sum_{c_0,c_1,\dots,c_{n-2}\in C}\prod_{r=1}^{n-1}\left(P^{E_{e_{n-1}-r},A_{a_{n-1}-\sum_{l=1}^r \delta_{c_{n-l}}}}\bigl(c_{n-1-r},c_{n-r}\bigr)\right)^{\delta_{c_{n-r}}}\cdot Q^{E_{e_{n-1}-r},A_{a_{n-1}-\sum_{l=1}^r \delta_{c_{n-l}}},c_{n-1-r}}(\delta_{c_{n-r}})\nonumber\\
&\quad\quad\quad\quad\quad\quad\quad\quad\quad\quad\quad\quad\times\pi_{c_0,e_{n-1}-(n-1),a_{n-1}-\sum_{l=0}^{n-2} \delta_{c_{l+1}}}\nonumber
\end{align}
Se sustituye \ref{eq2} en \ref{eq1}, tomando en cuenta que $e_{n-1}=e_n-1$ y $a_{n-1}=a_n-\delta_{c_n}$, y pasando el \'indice a $r-1$, se concluye que
\begin{align*}
\Prob\bigl[&(X_n,Y_n,Z_n)=(c_n,e_n,a_n)\bigr]\\
&=\sum_{c_{n-1}\in C}\left(P^{E_{e_n-1},A_{a_n-\delta_{c_n}}}\bigl(c_{n-1},c_n\bigr)\right)^{\delta_{c_n}}\cdot Q^{E_{e_n-1},A_{a_n-\delta_{c_n}},c_{n-1}}\bigl(\delta_{c_n}\bigr)\\
&\quad\quad \times\hspace{-.5cm}\sum_{c_0,c_1,\dots,c_{n-2}\in C}\prod_{r=2}^n\left(P^{E_{e_n-r},A_{a_n-\sum_{l=1}^r \delta_{c_{n+1-l}}}}\bigl(c_{n-r},c_{n+1-r}\bigr)\right)^{\delta_{c_{n+1-r}}}\cdot Q^{E_{e_n-r},A_{a_n-\sum_{l=1}^r \delta_{c_{n+1-l}}},c_{n-r}}(\delta_{c_{n+1-r}})\\
&\hspace{4cm}\times\pi_{c_0,e_n-n,a_n-\sum_{l=0}^{n-1} \delta_{c_{l+1}}}\\
&=\sum_{c_0,c_1,\dots,c_{n-1}\in C}\prod_{r=1}^n\left(P^{E_{e_n-r},A_{a_n-\sum_{l=1}^r \delta_{c_{n+1-l}}}}\bigl(c_{n-r},c_{n+1-r}\bigr)\right)^{\delta_{c_{n+1-r}}}\cdot Q^{E_{e_n-r},A_{a_n-\sum_{l=1}^r \delta_{c_{n+1-l}}},c_{n-r}}(\delta_{c_{n+1-r}})\\
&\hspace{4cm}\times\pi_{c_0,e_n-n,a_n-\sum_{l=0}^{n-1} \delta_{c_{l+1}}}.
\end{align*}
\end{proof}

\bigskip

\begin{corollary}\label{probrangos}
La probabilidad de que una persona, en el a$\tilde{\text{n}}$o $n$, se encuentre en la categor\'ia $c_n$, dentro de los rangos de ``edad'' y ``antig$\ddot{\text{u}}$edad'', $E_n\in\left\{E^{(0)},E^{(1)},E^{(2)},E^{(3)},E^{(4)}\right\}$ y $A_n\in\left\{A^{(1)},A^{(2)},A^{(3)},A^{(4)}\right\}$, respectivamente, ser\'ia
\begin{align*}
\Prob\left[(X_n,Y_n,Z_n)\in\{c_n\}\times E_n\times A_n\right]=\sum_{\substack{e_n\in E_n\\a_n\in A_n}}\Prob\left[(X_n,Y_n,Z_n)=(c_n,e_n,a_n)\right]
\end{align*}
\end{corollary}

\bigskip

\begin{corollary}\label{poblacion}
De una poblaci\'on inicial $I_0$, la cantidad esperada de personas $\Mean\bigl[I_n^{c_n,E_n,A_n}\bigr]$, en la categor\'ia $c_n\in C$, dentro de los rangos de edad $E_n\in\left\{E^{(0)},E^{(1)},E^{(2)},E^{(3)},E^{(4)}\right\}$ y de antig$\ddot{\text{u}}$edad $A_n\in\left\{A^{(1)},A^{(2)},A^{(3)},A^{(4)}\right\}$, para el $n-$\'esimo a$\tilde{\text{n}}$o ser\'ia
$$\Mean\bigl[I_n^{c_n,E_n,A_n}\bigr]=I_0\cdot\Prob\bigl[(X_n,Y_n,Z_n)\in\{c_n\}\times E_n\times A_n\bigr].$$
\end{corollary}

\subsection{Distribuci\'on de las Caracter\'isticas}

Sea $(W_1,W_2,\dots,W_{N_S})$ el vector aleatorio de caracter\'isticas.

\bigskip

{\bf Hip\'otesis:} Asumimos el vector aleatorio $(W_1,W_2,\dots,W_{N_S})$ es estacionario, por lo que podemos considerarlo independiente del tiempo. A\'un m\'as, se asume que, para $\{c_0\}\times E_0\times A_0\subset C\times E\times A$,
\begin{align*}
\Prob&\bigl[W_i=S^{j_i}_i,i=1,\dots,N_S\ \bigr/\ (X_n,Y_n,Z_n)\in\{c_0\}\times E_0\times A_0\bigr]\\
&\quad\quad\quad\quad=\Prob\bigl[W_i=S^{j_i}_i,i=1,\dots,N_S\ \bigr/\ (X_0,Y_0,Z_0)\in\{c_0\}\times E_0\times A_0\bigr]
\end{align*}

\smallskip

Defina la probabilidad estacionaria de la distribuci\'on de la poblaci\'on respecto a las caracter\'isticas, como
\begin{align}\label{distcaract}
R^{c_0,E_0,A_0}\bigl(S^{j_1}_1,\dots,S^{j_{N_S}}_{N_S}\bigr):=\Prob\bigl[W_i=S^{j_i}_i,i=1,\dots,N_S\ \bigr/\ (X_0,Y_0,Z_0)\in\{c_0\}\times E_0\times A_0\bigr]
\end{align}

\smallskip

\begin{corollary}\label{poblacioncar}
De una poblaci\'on inicial $I_0$, la cantidad esperada de personas $\Mean\bigl[I_n^{c_n,E_n,A_n}\bigl(S^{j_1}_1,\dots,S^{j_{N_S}}_{N_S}\bigr)\bigr]$, en la categor\'ia $c_n\in C$, dentro de los rangos de edad y antig$\ddot{\text{u}}$edad, $E_n\in\left\{E^{(0)},E^{(1)},E^{(2)},E^{(3)},E^{(4)}\right\}$ y $A_n\in\left\{A^{(1)},A^{(2)},A^{(3)},A^{(4)}\right\}$, respectivamente, para el $n-$\'esimo a$\tilde{\text{n}}$o, y con las caracter\'isticas $S^{(j_i)}_i$, para $i=1,\dots,N_S$, ser\'ia
$$\Mean\bigl[I_n^{c_n,E_n,A_n}\bigl(S^{j_1}_1,\dots,S^{j_{N_S}}_{N_S}\bigr)\bigr]=I_0\cdot R^{c_n,E_n,A_n}\bigl(S^{j_1}_1,\dots,S^{j_{N_S}}_{N_S}\bigr)\cdot\Prob\bigl[(X_n,Y_n,Z_n)\in\{c_n\}\times E_n\times A_n\bigr].$$
\end{corollary}

\section{Ajuste del Modelo}\label{ajusteprob}

Se considera una base hist\'orica donde cada individuo es registrado mes a mes, con su categor\'ia, su edad, y su antig$\ddot{\text{u}}$edad; as\'i como las caracetr\'isticas del mismo. Se trabajar\'a el escenario de ajuste con categor\'ias mensual, puesto que es m\'as probable que falten datos a que sobren, sin embargo, el ajuste con datos anuales se sigue de manera sencilla. Adem\'as, se asume que solo se tiene informaci\'on de los individuos dentro del sistema, mientras que la informaci\'on de las personas fuera del sistema no se posee.

\smallskip

\subsection{Poblaci\'on Total/Reserva}

Se toma una poblaci\'on total de $I_0$ personas,\footnote{Se puede tomar como la PEA para un fondo de pensiones como el IVM, o como un n\'umero definido en el caso de una instituci\'on que no posee un n\'umero representativo de la PEA.} distribuidas de la siguiente manera:
\begin{itemize}

\item $N_m^{C^{(r)},e,a}$ es la cantidad de personas que en el mes $m$ se encontraban en la categor\'ia $C^{(r)}$, con $e$ \anos de edad, y $a$ \anos de antig$\ddot{\text{u}}$edad, para $r=1,\dots,N_C$, y $e\in E$, $a\in A$.

\item Se asume que existen $N_e$ personas de edad $e\in[E_l,E_u]$, es decir, hay $N_{E_l}$ personas con edad $E_l$ a\~nos, $N_{E_l+1}$ personas con edad $E_l+1$ a\~nos, y as\'i sucesivamente hasta llegar a $N_{E_u}$ personas con edad $E_u$ a\~nos. De este modo, el n\'umero de personas con edad $e$ a\~nos que est\'an en la categor\'ia $C^{(0)}$ durante el mes $m$ ser\'ian
$$N_m^{C^{(0)},e}=N_e-\sum_{r=1}^{N_C}\sum_{a\in A}N_m^{C^{(r)},e,a}.$$

\item Esta poblaci\'on ``reserva'' de $e$ \anos de edad, que se encuentra en la categor\'ia $C^{(0)}$, se asume uniformemente distribuida entre antig$\ddot{\text{u}}$edades, es decir, se divide la poblaci\'on $N_m^{C^{(0)},e}$ en partes iguales, y dicha cantidad corresponder\'ia a $N_m^{C^{(0)},e,a}$ para los valores de $a$ posibles dada la edad $e$. Por el contrario, $N_m^{C^{(0)},e,a}=0$ para las antig\"uedades $a$ que son imposibles con la edad $e$.

\end{itemize}

\subsection{Distribuci\'on Inicial}

Sea $\mathcal{M}$ el conjunto de meses observados, donde $m=0$ para el \'ultimo mes observado, $m=-1$ para el pen\'ultimo mes observado, y as\'i sucesivamente, hasta llegar a $m=-M$ que corresponder\'ia al mes m\'as antiguo observado. Se emplea la contabilizaci\'on anterior, tomando $N_0^{C^{(r)},e,a}$ como el n\'umero promedio de personas en la categor\'ia $C^{(r)}$, con $e$ \anos de edad, y $a$ \anos de antig$\ddot{\text{u}}$edad, definido por
$$N_0^{C^{(r)},e,a}=\frac{1}{12}\sum_{m=-11}^0N_m^{C^{(r)},e,a}.$$
El promedio empleado es sobre el \'ultimo a\~no observado, el cual es la base para la proyecci\'on, pero puede hacerse para m\'as periodos.

\bigskip

Se toma la distribuci\'on inicial de la poblaci\'on como
$$\pi_{c_0,e_0,a_0}=\frac{N_0^{c_0,e_0,a_0}}{I_0}=\frac{\text{Poblaci\'on de categor\'ia $c_0$, edad $e_0$, y antig$\ddot{\text{u}}$edad $a_0$}}{\text{Poblaci\'on Total}},$$
para una tripleta $(c_0,e_0,a_0)\in C\times E\times A$.

\subsection{Probabilidades de Transici\'on Mensual}

Para el $m-$\'esimo mes, tome
$$N_m^{C^{(r)},E^{(i)},A^{(k)}}=\sum_{e\in E^{(i)},a\in A^{(k)}}N_m^{C^{(r)},e,a},$$
como el n\'umero total de personas que estaban en la categor\'ia $C^{(r)}$, con edad y antig$\ddot{\text{u}}$edad en los rangos $E^{(i)}$ y $A^{(k)}$, respectivamente, en el mes $m$.

\bigskip

Por otro lado, sea $N_m^{E^{(i)},A^{(k)}}\left(C^{(r)},C^{(l)}\right)$, el n\'umero de personas con edad y antig$\ddot{\text{u}}$edad en los rangos $E^{(i)}$ y $A^{(k)}$, respectivamente, que se encontraba en el mes $m$ en la categor\'ia $C^{(r)}$, despu\'es de un mes (en el mes $m+1$), se encontraba en la categor\'ia $C^{(l)}$.

\bigskip

La probabilidad de transici\'on mensual del $m-$\'esimo m\'es estar\'ia dada por\footnote{Observe que se resta del denominador el n\'umero de personas que llegaron a $C^{(0)}$, esto debido a que las transiciones se cuentan solamente para aquellos movimientos dentro del sistema. Sin embargo, no hay restricci\'on en que la persona pase de $C^{(0)}$ a una categor\'ia dentro del sistema.}
$$\widetilde{P}_m^{E^{(i)},A^{(k)}}\left(C^{(r)},C^{(l)}\right)=\frac{N_m^{E^{(i)},A^{(k)}}\left(C^{(r)},C^{(l)}\right)}{N_m^{C^{(r)},E^{(i)},A^{(k)}}-N_m^{E^{(i)},A^{(k)}}\left(C^{(r)},C^{(0)}\right)}.$$
El estimador de esta probabilidad ser\'ia
$$\widetilde{P}^{E^{(i)},A^{(k)}}\left(C^{(r)},C^{(l)}\right)=\frac{1}{|\mathcal{M}|}\sum_{m\in \mathcal{M}}\widetilde{P}_m^{E^{(i)},A^{(k)}}\left(C^{(r)},C^{(l)}\right),$$
donde $|\mathcal{M}|$ es el n\'umero de elementos que tiene $\mathcal{M}$. Adem\'as, seg\'un Zucchini \cite[p. 21]{Zucchini}, para una cadena de Markov este es un estimador insesgado de las probabilidades de transici\'on.

\subsection{Probabilidades de Transici\'on Anual}

Una vez estimadas las probabilidades de transici\'on mensuales, se obtienen las anuales con la f\'ormula \eqref{probtransmens}
$$P^{E^{(i)},A^{(k)}}\bigl(C^{(r)},C^{(l)}\bigr)=\sum_{t=1}^{12}\sum_{\widetilde{c}_1,\dots,\widetilde{c}_{t-1}\in C}\prod_{k=1}^{t}\widetilde{P}^{E^{(i)},A^{(k)}}\bigl(\widetilde{c}_k,\widetilde{c}_{k-1}\bigr)\cdot p^{(i)}_t,$$
con $\widetilde{c}_0=C^{(r)}$ y $\widetilde{c}_t=C^{(l)}$.

\subsection{Probabilidades de Ingreso al Sistema}

Sea $\mathcal{N}$ el conjunto de a\~nos observados. $N_n^{E^{(i)},A^{(k)},C^{(r)}}(1)$ el n\'umero de personas de la categor\'ia $C^{(r)}$, con edad y antig$\ddot{\text{u}}$edad en los rangos $E^{(i)}$ y $A^{(k)}$, respectivamente, que en el $n-$\'esimo a\ene o fueron contratados, y $N_n^{E^{(i)},A^{(k)},C^{(r)}}(0)$ aquellos con las mismas caracter\'isticas que no lo fueron. 

\bigskip

Las probabilidades de ingreso al sistema del $n-$\'esimo a\ene o estar\'ian dads por
$$Q_n^{E^{(i)},A^{(k)},C^{(r)}}(l)=\frac{N_n^{E^{(i)},A^{(k)},C^{(r)}}(l)}{N_n^{E^{(i)},A^{(k)},C^{(r)}}(0)+N_n^{E^{(i)},A^{(k)},C^{(r)}}(1)}.$$
El estimador de esta probabilidad ser\'ia
$$Q^{E^{(i)},A^{(k)},C^{(r)}}(l)=\frac{1}{|\mathcal{N}|}\sum_{n\in\mathcal{N}}Q_n^{E^{(i)},A^{(k)},C^{(r)}}(l).$$

\subsubsection{Distribuci\'on de la Poblaci\'on por Caracter\'isticas}

Tomamos el dato $N_m^{C^{(r)},E^{(i)},A^{(k)}}$ del n\'umero total de personas que en el $m-$\'esimo mes pertenec\'ian a la categor\'ia $C^{(r)}$, dentro del grupo de edad $E^{(i)}$ y rango de antig$\ddot{\text{u}}$edad $A^{(k)}$. De este grupo, obtenga el n\'umero $N_m^{C^{(r)},E^{(i)},A^{(k)}}(S^{j_1}_1,\dots,S^{j_{N_S}}_{N_S})$ de personas que ten\'ian las caracter\'isticas salariales $(S^{j_1}_1,\dots,S^{j_{N_S}}_{N_S})$, durante el mes $m$. Utilizando la f\'ormula \eqref{distcaract},
$$R_m^{C^{(r)},E^{(i)},A^{(k)}}\bigl(S^{j_1}_1,\dots,S^{j_{N_S}}_{N_S}\bigr)=\frac{N_m^{C^{(r)},E^{(i)},A^{(k)}}\bigl(S^{j_1}_1,\dots,S^{j_{N_S}}_{N_S}\bigr)}{N_m^{C^{(r)},E^{(i)},A^{(k)}}}.$$
El estimador de esta probabilidad ser\'ia
$$R^{C^{(r)},E^{(i)},A^{(k)}}\bigl(S^{j_1}_1,\dots,S^{j_{N_S}}_{N_S}\bigr)=\frac{1}{|\mathcal{M}|}\sum_{m\in\mathcal{M}}R_m^{C^{(r)},E^{(i)},A^{(k)}}\bigl(S^{j_1}_1,\dots,S^{j_{N_S}}_{N_S}\bigr).$$

\section{Simulaci\'on de Montecarlo}

Se desea proyectar la poblaci\'on utilizando una simulaci\'on de Montecarlo. Defina
$$P^n_{r,i,k}:=\Prob\left[(X_n,Y_n,Z_n)\in\{C^{(r)}\}\times E^{(i)}\times A^{(k)}\right];\quad \text{ y }\quad R^{r,i,k}_{j_1,\dots,j_{N_S}}:=R^{C^{(r)},E^{(i)},A^{(k)}}\bigl(S^{j_1}_1,\dots,S^{j_{N_S}}_{N_S}\bigr).$$

Luego, defina: 

\vspace{-.8cm}

$$ V^{n,i,r,k }_{j_1,\dots,j_{N_S}} := P^n_{r,i,k} R^{r,i,k}_{j_1,\dots,j_{N_S}}$$ 

Que representa la probabilidad de que una persona est\'e en la $\{N_S+3\}$-tupla dada por: ${i,r,k,j_1,\dots,j_{N_S}}$ en el a\~no $n$, por definici\'on de probabilidad condicional.

\subsection{Algoritmo}

Suponga que se van a proyectar $N$ a\~nos, utilizando la t\'ecnica de Montecarlo con 10.000 iteraciones por a\~no.

\smallskip

{\bf Pseudoc\'odigo:}
\begin{itemize}
\item{\bf Recibe:} Probabilidades $V^{n,i,r,k }_{j_1,\dots,j_{N_S}}$, y la poblaci\'on inicial $I_0$

\item Inicie $I_1^{C^{(r)},E^{(i)},A^{(k)}}\bigl(S^{j_1}_1,\dots,S^{j_{N_S}}_{N_S}\bigr)=0,\dots,I_{N}^{C^{(r)},E^{(i)},A^{(k)}}\bigl(S^{j_1}_1,\dots,S^{j_{N_S}}_{N_S}\bigr)=0$, para todo $r,i,k,j_1,\dots,j_{N_S}$.
\item Para $n=1,\dots N$ (proyecci\'on para $N$ a$\tilde{\text{n}}$os) genere el siguiente vector aleatorio
$$\left\{I_n^{C^{(r)},E^{(i)},A^{(k)}}\bigl(S^{j_1}_1,\dots,S^{j_{N_S}}_{N_S}\bigr)\right\}_{i,r,k,j_1,\dots,j_{N_S}} \sim Multinom\big(I_0;\{V^{n,i,r,k }_{j_1,\dots,j_{N_S}}\}_{i,r,k,j_1,\dots,j_{N_S}}\big).$$

\item {\bf Devuelve:} La poblaci\'on para cada categor\'ia, grupo de edad, rango de antig$\ddot{\text{u}}$edad, y cada a$\tilde{\text{n}}$o en el futuro, agregada y desagregada por caracter\'isticas, para los siguientes $N$ a$\tilde{\text{n}}$os.
\end{itemize}

\begin{figure}[ht]

\begin{center}
\includegraphics[width = 12cm]{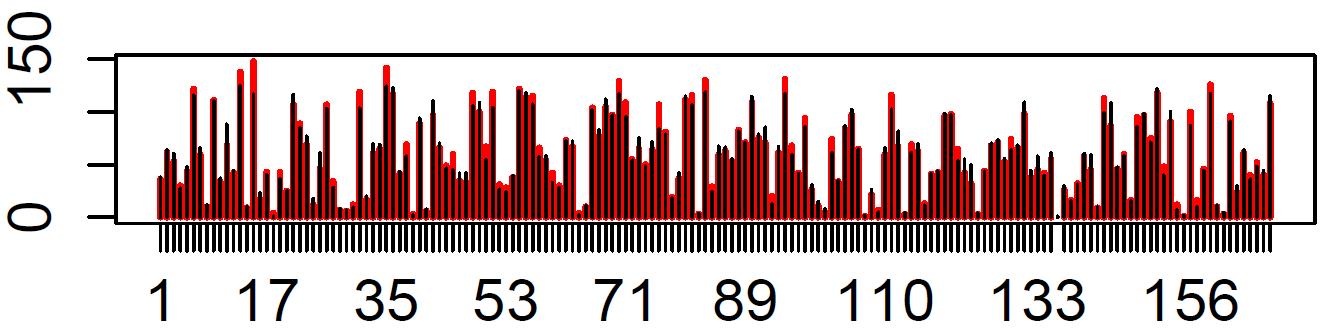}
\caption{Histograma de la Generaci\'on (Rojo) Vs Probabilidad Estimadas (Negro)}
\end{center}

\end{figure}

\section{Ejemplo de C\'alculo de Gastos por Remuneraciones}

Se utilizaron los datos de planillas de una instituci\'on p\'ublica para los a\~nos 2004-2015. Las proyecciones se har\'an por un rango de 25 a\~nos. Adem\'as, se consideran los grupos de edad de la siguiente forma:
\begin{itemize}

\item Grupo $E^{(0)}$: Personas menores a 18 a\~nos, considerados como potenciales empleados al cumplir 18 a\~nos.

\item Grupo $E^{(1)}$: Aquellas personas con edades entre los 18 a\~nos y los 30 a\~nos (no cumplidos).

\item Grupo $E^{(2)}$: Individuos con edades entre los 30 a\~nos y hasta los 40 a\~nos (no cumplidos).

\item Grupo $E^{(3)}$: Compuesto por personas con edades entre los 40 a\~nos y los 50 a\~nos (no cumplidos).

\item Grupo $E^{(4)}$: Este grupo es considerado como el de las potenciales jubilaciones, pues est\'a integrados por individuos con edades superiores a los 50 a\~nos.

\end{itemize}
Del mismo modo, se consideran cuatro grupos de antig\"uedades (indicador de su ligamen con la instituci\'on):
\begin{itemize}

\item Grupo $A^{(1)}$: Personas con menos de 15 a\~nos de trabajar en la instituci\'on (no necesariamente consecutivos).

\item Grupo $A^{(2)}$: Individuos que han trabajado para la entidad entre 15 y 30 a\~nos (no necesariamente consecutivos).

\item Grupo $A^{(3)}$: Integrado por trabajadores que han laborado entre 30 y 45 a\~nos (no necesariamente consecutivos).

\item Grupo $A^{(4)}$: Compuesto por personas con un nivel elevado de relaci\'on con la instituci\'on, habiendo laborado por m\'as de 45 a\~nos (no necesariamente consecutivos).

\end{itemize}
La tripleta, compuesta por categor\'ia, grupo de edad y antig\"uedad, utiliza las categor\'ias salariales dadas por:
\begin{center}
$C^{(1)}$ = 11;\quad $C^{(2)}$ = 12;\quad $C^{(3)}$ = 13;\quad $C^{(4)}$ = 14;\quad $C^{(5)}$ = 21;\quad $C^{(6)}$ = 22;\quad $C^{(7)}$ = 23;\quad$C^{(8)}$ = 24;\\
$C^{(9)}$ = 25;\quad $C^{(10)}$ = 31;\quad $C^{(11)}$ = 32;\quad $C^{(12)}$ = 34;\quad $C^{(13)}$ = 35;\quad $C^{(14)}$ = 36;\quad $C^{(15)}$ = 37;\quad $C^{(16)}$ = 38;\\
$C^{(17)}$ = 41;\quad $C^{(18)}$ = 42;\quad $C^{(19)}$ = 43;\quad $C^{(20)}$ = 44;\quad $C^{(21)}$ = 45;\quad $C^{(22)}$ = 47;\quad $C^{(23)}$ = 48;\quad $C^{(24)}$ = 49;\\
$C^{(25)}$ = 53;\quad $C^{(26)}$ = 54;\quad $C^{(27)}$ = 57;\quad $C^{(28)}$ = 63;\quad $C^{(29)}$ = 79;\quad $C^{(30)}$ = 82;\quad $C^{(31)}$ = 84;\quad $C^{(32)}$ = 86;\\
$C^{(33)}$ = 87;\quad $C^{(34)}$ = 88;\quad $C^{(35)}$ = 89;\quad $C^{(36)}$ = 90;\quad $C^{(37)}$ = 91.
\end{center}
Por su parte, se consideraron como ``caracter\'isticas'' solamente las siguientes: anualidad, dedicaci\'on exclusiva, prohibici\'on, disponibilidad, r\'egimen de pensiones, y jornada laboral. Esto por cuanto los gastos derivados de los cuatro primeros representan el 90\% de los gastos en pluses salariales, el quinto se emplea para calcular el correcto monto gastado por la universidad en aportes patronales a pensiones, as\'i como para poder efectuar el respectivo analisis sobre el efecto del r\'egimen sobre las jubilaciones, y el \'ultimo debido a que no todos los empleados trabajan en jornadas de 40 horas semanales (algunos trabajan m\'as y otros trabajan menos horas).

\bigskip

Para un trabajador de categor\'ia salarial $C^{(i)}$ (asociada un\'ivocamente a un salario base $W_i$ seg\'un la escala salarial del segundo semestre del 2015), con jornada laboral $J$, y porcentajes\footnote{Estos porcentajes pueden ser 0\%.} de anualidad $A$, dedicaci\'on exclusiva $DX$, prohibici\'on $P$ y disponibilidad $D$, se calcula el gasto anual en el a\~no $N$ por concepto de sus salarios sin garant\'ias sociales $G$ con la siguiente f\'ormula:
\begin{align*}
G&(N,J,W_i,A,DX,P,D)\\
&=\frac{J}{40} \cdot \left[6\cdot W_i\cdot \left((1+3,88\%)^{(N - 2016)+1/2} + (1+3,88\%)^{N - 2015}\right)\right]\cdot \frac{(1+ A + DX + P + D)}{0.90}.
\end{align*}
El primer t\'ermino $(J/40)$ nos indica el porcentaje del salario que recibe por la proporci\'on de horas que trabaja respecto de la jornada completa. El segundo t\'ermino nos contabiliza los 12 salarios base, tomando en cuenta que hay 6 salarios correspondientes a $2\cdot(N - 2016)+1$ aumentos\footnote{Considerados para efectos de este estudio como iguales a la inflaci\'on esperada del 3,88\%.} semestrales desde diciembre del 2015, y otros 6 salarios con $(N - 2015)$ aumentos anuales desde diciembre del 2015. El \'ultimo t\'ermino nos presenta lo gastado por estos 12 salarios, m\'as su porcentaje de anualidad, dedicaci\'on exclusiva, prohibici\'on, y disponibilidad. Como estos no son todos los pluses salariales, y dado que la suma de los pluses que no est\'an siendo considerados representan el 10\% de los gastos totales en salarios, se procede a normalizar el monto calculado por el factor $\frac{1}{0,90}$, estimando as\'i el verdadero gasto anual para este trabajador.

\bigskip

Del mismo modo, para un trabajador con las caracter\'isticas anteriormente indicadas, aunado al porcentaje de cotizaci\'on de r\'egimen de pensiones $R$, el porcentaje de salario escolar $E$, 14,25\% de garant\'ias sociales (Seguro de Enfermedad y Maternidad, Ley de Protecci\'on al Trabajador, Banco Popular)\footnote{La instituci\'on est\'a excenta de pagar IMAS, Asignaciones Familiares e INA.}, 4,25\% para el Fondo de Cesant\'ia y Asociaciones, 8,33\% de Aguinaldo, y 0,25\% por Riesgos de Trabajo del INS, se utiliza la siguiente f\'ormula para estimar el gasto total anual $GT$ en este empleado:
\begin{align*}
GT&(N,J,W_i,A,DX,P,D,R,E)\\
&=G(N,J,W_i,A,DX,P,D)\cdot(1+ E)\cdot\bigl((1 + R + 14,25\% + 4,25\%) + 8,33\%\bigr)*(1 + 0,25\%).
\end{align*}
Seg\'un directiva presidencial, el salario escolar ser\'a incrementado en tractos durante los a\~nos 2016-2018 hasta alcanzar el valor de un salario completo (como el aguinaldo). Esto fue tomado en cuenta, de modo que los porcentajes que se utilizaron para $E$ fueron de 8,19\% hasta el 2015, pasando a 8,23\% en el 2016, 8,28\% en el 2017, y manteni\'endose en un nivel de 8,33\% a partir del 2018. Por su parte, se utiliz\'o un porcentaje de aporte patronal de 6,75\% si el empleado est\'a en el llamado ``Capitalizaci\'on'' (de JUPEMA), de 5\% si est\'a en el de ``Reparto'' (de JUPEMA), y se sigui\'o el Transitorio XI de la CCSS para los empleados en el IVM, es decir, se tomaron aportes patronales de 4.92\% en el 2014, de 5.08\% entre los a\~nos 2015-2019, de 5.25\% en el rango 2020-2024, de 5.42\% entre el 2025 y el 2029, de 5.58\% para el quinquenio 2030-2034, y de 5.75\% a partir del 2035.

\bigskip

Se considera que los tiempos de parada $T_i$ se distribuyen uniformemente, es decir, $p^{(i)}_t=\frac{1}{12}$ para todo $i$. Adem\'as, las probabilidades de transici\'on se ponderaron por sus jornadas, de modo que cada trabajador aporta a las transiciones dependiendo de las horas que laboran en la instituci\'on.

\bigskip

Para medir la calidad del ajuste, as\'i como intentar cuantificar la certeza de sus proyecciones, se implement\'o el m\'etodo conocido como ``backtesting'', el cual consiste en utilizar los datos de los a\~nos 2004-2013 para ajustar los par\'ametros, y contrastar los valores observados en los a\~nos 2014-2015 contra los proyectados por el modelo. Inicialmente notamos que las predicciones globales son bastante buenas, donde el gasto total (de las cuentas descritas por este modelo) cuantificaron 64.215.040.730 colones en el 2014 y 70.044.868.080 colones en el 2015, mientras que lo esperado\footnote{Nos referimos a la ``Esperanza Matem\'atica'' de la proyecci\'on.} por el modelo rondaba los 64.149.862.676 colones para el 2014 y 69.262.303.902 colones para el 2015. Se estar\'ia contando con un error de 0,1\% en el 2014 y de 1,12\% para el 2015.

\bigskip

Se comprueba que el nivel de predicci\'on del modelo respecto a los valores observados (tanto en n\'umero de trabajadores como en colones gastado en sus salarios) es muy alta, validando la confiabilidad de las conclusiones que de este modelo se deriven. Al final del documento (secci\'on de anexos) se presentan los gr\'aficos sobre poblaci\'on observada versus esperada, as\'i como los gastos calculados contra los observados, tanto para las tres categor\'ias m\'as importantes, as\'i como para los tres grupos de edad m\'as relevantes.

\subsection{Conclusiones}

Se comprob\'o, con datos reales, c\'omo el modelo logra ajustar y proyectar, de manera confiable, los valores futuros tanto demogr\'aficos como financieros. Adem\'as, representa una herramienta \'util para trabajar el problema de proyectar poblaciones para fondos de pensiones o para instituciones p\'ublicas. El modelo presenta un reto de programaci\'on, pero una vez hecha esta inversi\'on, el mismo posee suficiente flexibilidad como para poder efectuar an\'alisis demogr\'aficos (como poblaciones en v\'ias de jubilaci\'on) as\'i como financieros (montos cotizados por los trabajadores a lo largo de la proyecci\'on). Se pueden incluir comportamientos futuros, como aumentos en probabilidades de transici\'on en ciertas categor\'ias, o aumentos/reducciones en pluses salariales.

\end{document}